\documentclass[10pt,twocolumn,twoside]{IEEEtran} 

\IEEEoverridecommandlockouts                              

\usepackage{graphics} 
\usepackage{epsfig} 
\usepackage{amsmath} 

\usepackage{amssymb,amsthm}  
\usepackage{amscd}
\usepackage[noadjust]{cite}
\usepackage{float}
\usepackage{times}
\usepackage{color}

\newtheorem{Definition}{Definition}
\newtheorem{Assumption}{Assumption}

\newtheorem{theorem}{Theorem}
\newtheorem{lem}{Lemma}
\newtheorem{pro}{Proposition}
\newtheorem{cor}{Corollary}

\newtheorem{Remark}{Remark}

\newcommand{\bb}{\mathbb}

\newcommand{\diag}{\operatorname{diag}}
\newcommand{\Sp}{\operatorname{Sp}}
\newcommand{\cl}{\operatorname{cl}}

\renewcommand{\)}{\right )}
\newcommand{\R}{\mathbb{R}}
\title{\LARGE \bf
Distributed Evaluation and Convergence of \\Self-Appraisals in Social Networks}

\author{Xudong Chen, Ji Liu, M.-A. Belabbas, Zhi Xu, Tamer Ba\c sar
\thanks{Xudong Chen, Ji Liu, M.-A. Belabbas and Tamer Ba\c sar are with the Coordinated Science Laboratory, University of Illinois at Urbana-Champaign, emails:
        \{xdchen, jiliu, belabbas, basar1\}@illinois.edu.
        }
\thanks{Zhi Xu is with the Department of Electrical Engineering and Computer Science, Massachusetts Institute of Technology, email:
        zhixu@mit.edu.
       }       
\thanks{This research was supported in part by the U.S. Air Force Office of Scientific Research (AFOSR) MURI grant FA9550-10-1-0573.}%
}

\begin{document}

\maketitle
\thispagestyle{empty}
\pagestyle{empty}

\begin{abstract}
We consider in this paper a networked system of opinion dynamics in continuous time, where the agents are able to evaluate their self-appraisals in a distributed way. In the model we formulate, the underlying network topology is described by a rooted digraph. For each ordered pair of agents $(i,j)$, we assign a function of self-appraisal to agent~$i$, which measures the level of importance of agent~$i$ to agent~$j$. Thus, by communicating only with her neighbors, each agent is able to calculate the difference between her level of importance to others and others' level of importance to her. The dynamical system of self-appraisals is then designed to drive these differences to zero. We show that for almost all initial conditions, the trajectory generated by this dynamical system asymptotically converges to an equilibrium point which is exponentially stable.
\end{abstract}

\section{Introduction}

Social network is a social structure made up of actors, such as agents and organizations, and the relationships between these actors, particularly those that are neighbors to each other.
The concept of social networks is familiar to most people because of the emergence of online social networking services such as Facebook, Twitter, and Google+. Many social behaviors spread through social networks of interacting agents.
Examples are opinion dynamics \cite{DeGroot}, adoption of new technology or products \cite{rogers}, voting \cite{kuran1}, and demonstrations \cite{Morse2012}.
In this paper, we will focus on an important issue in opinion dynamics
which is about how each agent in the social network evaluates her importance in a distributed way.

Over the past decades, there has been considerable attention paid to understanding
how an agent's opinion evolves over time.
In social science, various models have been proposed to illustrate opinion dynamics.
Notable among them are the four classical models: the DeGroot model~\cite{DeGroot}, the Friedkin-Johnsen model~\cite{Johnsen}, the Hegselmann-Krause model~\cite{Krause}, and the Deffuant-Weisbuch model~\cite{dw}.
In the DeGroot model, each agent has a fixed set of neighbors and takes a convex combination of her own opinion and the opinions of her neighbors. The Friedkin-Johnsen model is a variation of the DeGroot model in which each agent  adheres to her initial opinion to a certain degree, which can be interpreted as her level of ``stubbornness". The Hegselmann-Krause model and the Deffuant-Weisbuch model define the neighbor sets in a different way: each agent takes a set of agents as her neighbors only if the opinions of these agents differ from her by no more than a certain confidence level. 
With the defined set of neighbors, each agent takes a convex combination of her own 
opinion with either (i) all of her current neighbors in the Hegselmann-Krause model, or (ii) only one of her current neighbors 
in the Deffuant-Weisbuch model.
Other important opinion models are, for example,  the Sznajd model~\cite{sznajd}, which uses the Ising spin model to describe a mechanism of making a decision in a closed community, and the voter model, which is a continuous-time Markov process defined over a lattice of integers, proposed by~\cite{voter1,voter2}.  

Recently, with the rapid expansion of large-scale, online networks,
there has been an increased interest in the analysis of opinion formation, with the objective of extending the classical models by taking into account more factors of social interactions
\cite{Acemoglu,yildiz,srikant,John,Bullo2,etesami2014game,Bullo3,Bullo1,Bullo5}.
In the work of \cite{Acemoglu} and \cite{yildiz}, the effects of the existence of stubborn agents--the agents who never update their opinions--are investigated in a randomized pair-wise updating process.
In \cite{srikant}, the opinion formation process is reformulated into a local interaction game, and the concept of stubbornness of an agent regarding her initial opinion is introduced.
The Krause model and its variations are studied in~\cite{John,Bullo2,etesami2014game,Bullo3}.
For example, a game-theoretic analysis of the Krause model is studied in \cite{etesami2014game}.
The work of \cite{Bullo3} takes into account exogenous factors, such as the influence of media, and
assumes that
each agent updates her opinion via the opinions of the population inside her ``confidence range'' and the information from an exogenous input in that range.
In the literature,
both discrete-time \cite{DeGroot,Johnsen} and continuous-time \cite{blondel,hendrickx} approaches have been adopted to model the update rule of opinions of agents.

Recently, Jia {\em et al} proposed the so-called DeGroot-Friedkin model~\cite{Bullo1,Bullo5}. This model uses the concept of reflected appraisal from sociology \cite{friedkin_reflected,reflected}, and studies the evolution of self-confidence, i.e., how confident an agent is in her opinions on a variety of issues.
Briefly speaking, reflected appraisal describes the phenomenon that agents' self-appraisals 
are influenced by the appraisals of other agents on them.
Following the work of \cite{Bullo1,Bullo5}, a modified DeGroot-Friedkin model is proposed in \cite{acc}
in which each agent updates her level of self-confidence in a more frequent manner.
Specifically, all the agents in the network update their own levels of self-confidence immediately after each time of discussion, instead of waiting for the opinion process to reach a consensus on any particular issue, which generally takes infinite time.
The analysis of the modified DeGroot-Friedkin model is carried out only in the special case
when the so-called {\it relative interaction matrix} is doubly stochastic.
Yet, a complete understanding of the system behavior for the most general case has remained open.

We introduce in this paper a continuous-time self-appraisal model whereby the agents in a social network are able to evaluate their self-appraisals in a distributed way.
For each ordered pair of agents $i$ and $j$, we assign a function of self-appraisal to
agent~$i$, which measures the level of importance of agent~$i$ to agent~$j$.
With local interaction with her neighbors only, each agent is able to calculate the difference between her level of
importance to others and others' level of importance to herself.
The proposed dynamical system of self-appraisals aims to drive these differences to zero.

{\color{black}\color{black}
Although many opinion dynamics are built on discrete-time scales,
a continuous-time model holds its merits in many ways: 
(i) a continuous-time model would be a natural choice if the  opinions of individuals evolve gradually over time, and moreover, can be used to describe the limiting behavior of a discrete-time system if the exchange of the opinions are frequent enough;  
(ii) often, the analysis of a continuous-time system would be easier to carried out than a discrete-time system;  
and (iii)  a complete analysis of a continuous-time system provides valuable insight into 
the dynamics of the discrete-time counterpart. We should note that the asymptotic behavior of the trajectories of 
the continuous-time self-appraisal model studied in this paper matches
numerical simulations of the discrete-time modified DeGroot-Friedkin model,
where the latter does not permit a thorough analysis.
}

We note here that the analysis of the self-appraisal model has been carried out earlier in our conference paper~\cite{chen2015distributed}, but for the case when the network topology is a strongly connected graph, for which the convergence of the corresponding dynamical system has been shown, with several of proofs left out.  In this paper, we generalize the result to the case when the network topology is a rooted graph, and show that under some mild assumptions, for almost all initial conditions, the trajectory generated by  the dynamical system asymptotically
converges to an equilibrium point which is exponentially stable. We provide a complete analysis, as well as proofs for establishing this result.

The remainder of the paper is organized as follows. In section~II, we describe in detail the self-appraisal model as well as the motivation behind it.
We also state the main theorem of the paper on the convergence of the self-appraisal model. In particular, the main theorem says that there is only one stable equilibrium point in the unit simplex (as the underlying space of the dynamical system), and moreover, for almost all initial conditions in the unit simplex, the trajectories of the dynamical system converge to that stable equilibrium point. Sections~III-V are devoted to establishing several properties that are needed for the proof of the main theorem. A detailed organization of these three sections will be given after the statement of Theorem~\ref{MAIN}.  We provide conclusions in the last section. The paper ends with an Appendix.

\section{The Self-Appraisal Model and Main Result}

In this section, we  introduce the continuous-time self-appraisal model, and state the main result of this paper.

\subsection{Background and notations}

 By convention, the neighbor relation among the $n$ agents in the network is  characterized by a directed graph (or digraph) $G = (V,E)$, with $V = \{1,\ldots, n\}$ the vertex set and $E$ the edge set.  
 We consider in this paper only  \emph{simple}  directed graphs, that is directed graphs with no self loops, and with at most one edge between each ordered pair of vertices. 
 Denote by $i\to j$ a directed edge of $G$ in which $i$ is the start-vertex and $j$ is the end-vertex,  and we say that $j$ is an {\it outgoing neighbor} of  $i$ and $i$ is an {\it incoming neighbor} of $j$. Denote by $V^+_i$ and $V^-_i$ the sets of incoming and outgoing neighbors of $i$, respectively.  A {\it directed path} (or simply path) of $G$ is a sequence of edges 
$$
i_1\to i_2 \to \ldots \to i_k 
$$ 
connecting vertices of $G$, and all the vertices in the path are distinct from each other.     
 A vertex~$i$ is said to be a {\bf root} if for any other vertex~$j$, there is a directed path from $j$ to $i$ in $G$. 
The digraph $G$ is said to be {\bf rooted} if it contains at least a root, and {\bf strongly connected} if each vertex of $G$ is a root. 
For a subset $V' \subset V$, we call $G'$ a subgraph of $G$ \emph{induced} by $V'$ if $G' = (V',E')$ and $E'$ contains any edge of  $E$ whose start-vertex and end-vertex are in $V'$. 
Denote by $V_r\subseteq V$ the collection of roots of $G$, and denote by $G_r = (V_r,E_r)$ the subgraph of $G$ induced by $V_r$. It is well known that the digraph $G_r$ is strongly connected. 


Denote by $\operatorname{Sp}[V]$  the $(n-1)$-simplex contained in $\mathbb{R}^n$ with vertices the standard basis vectors $e_1, \ldots, e_n \in \mathbb{R}^n$.  For $V' \subset V$, we define similarly $\operatorname{Sp}[V']$ as the convex hull of $\{e_i \mid i \in V'\}$, i.e., 
$$\operatorname{Sp}[V'] := \left\{ \sum_{i \in V'} \alpha_i e_i \mid \alpha_i \geq 0, \sum_{i \in V'} \alpha_i = 1 \right\}.$$
A point $ x\in \Sp[V]$ is said to be a {\bf boundary point} of $\Sp[V]$ if $x_i = 0$ for some $i\in V$. 
If $V'$ is a proper subset of $V$, then each $x\in \Sp[V']$ is a boundary point of $\Sp[V]$. 

For an arbitrary dynamical system $\dot{x} = f(x)$ in $\bb{R}^n$, a subset $Q\subseteq\bb{R}^n$ is said to be {\bf positive invariant} if for any initial condition $x(0)\in Q$, the solution $x(t)$ of the dynamical system is contained in $Q$ for all time $t\ge 0$.     


\subsection{The self-appraisal model}


To introduce the self-appraisal model, we first consider the following opinion consensus process, as a variation of the DeGroot model~\cite{DeGroot}, in continuous-time: Let $G= (V,E)$ be a rooted graph; each vertex $i$ of $G$ has at least one outgoing neighbor. The opinion consensus process is described by
\begin{equation}\label{artif}
\dot z_i(t) 
            = (1-x_i(t)) \left (-z_i(t) + \sum_{j\in V_i^-}c_{ij}z_j(t)\right).
\end{equation}
Each $z_i(t)$ is a real number (or vector), representing the opinion of agent~$i$ on certain ongoing issue(s) at time $t$. The  number $ x_i(t)\in [0,1]$ represents the current self-appraisal of agent~$i$ in the social network. The coefficients $c_{ij}$'s,  termed as {\it relative inter-personal weights} in~\cite{Bullo1}, are positive real numbers, satisfying the following condition:
  \begin{equation}\label{eq:summationtoone}
\sum_{j\in V^-_i} c_{ij} =1, \hspace{10pt} \forall i\in V.
\end{equation}
{\color{black} Each $c_{ij}$ can be set by agent~$i$ herself, encoding the willingness of agent~$i$ to accept the opinion of agent~$j$. Alternatively, $c_{ij}$ can be also viewed as a measure of the influence of the opinion of agent~$j$ on the opinion of agent~$i$.  }

We then recognize that system~\eqref{artif}, with condition~\eqref{eq:summationtoone}, is a continuous-time consensus process~\cite{luc04}, with the dynamics of~$z_i(t)$ 
scaled by the non-negative factor $(1-x_i(t))$.  
Note that in~\eqref{artif}, a larger value of $x_i(t)$ implies a smaller value of $\|\dot{z}_i(t)\|$. So then, $(1 - x_i(t))$ can be viewed as a measure of the total amount of opinions agent~$i$ accepts from others at time~$t$, and $c_{ij}(1 - x_i(t))$ is the corresponding portion agent~$i$  accepts from agent~$j$. 

With the opinion consensus-process~\eqref{artif} in mind, we propose the following dynamics for the evolutions of self-appraisals:
\begin{equation}\label{MODEL}
\dot x_i = -(1-x_i)x_i + \sum_{j\in V_i^+} c_{ji}(1-x_j)x_j, \hspace{10pt} \forall i\in V.  
\end{equation}
Similar to the works~\cite{Bullo1,acc}, the self-appraisals are scaled so that they sum to one, i.e., 
$$
x := (x_1,\ldots, x_n)\in \Sp[V].
$$
Note that such a scaling will be effective along the evolution; indeed, we show in the next section that $\Sp[V]$ is a positive invariant set for system~\eqref{MODEL}. 

We now   justify the self-appraisal model introduced above. From~\eqref{MODEL}, the evolution of $x_i$ is determined by two terms: $\sum_{j\in V_i^+} c_{ji}(1-x_j)x_j$ and  $(1-x_i)x_i$.  Explanations of these two terms are given below.

Each summand $c_{ji}(1-x_j)x_j$ in the first term is a product of two factors: one is $c_{ji}(1 - x_j)$ which measures the amount of opinion agent~$j$ accepts from agent~$i$ during the consensus process~\eqref{artif}, and the other factor $x_j$ is the self-appraisal of agent~$j$ which reflects the importance of agent $j$ in the network. Thus, their product $c_{ji}(1-x_j)x_j$ can be viewed as the measure of importance of agent~$i$ to agent~$j$.   
Of course, there are numerous ways of modeling these two factors, 
 the rationale behind the choice of using the product is given below.

We first consider the case where $x_j=0$; then agent~$j$ is not important at all in the social network, and thus, agent~$i$ will not increase her self-appraisal regardless of how much opinion  agent~$j$ accepts from her. On the other hand, consider the case where $x_j=1$; then from the consensus process, we see that agent~$j$ will not accept any opinion from agent~$i$. Thus, in this situation,  agent~$i$ will not increase her self-appraisal either regardless of how important agent~$j$ is in the network.  By taking these two cases into account, we  realize  that   $c_{ji}(1-x_j)x_j$ may be one of the simplest expressions that realistically  models how the neighbor $j$ affects the self-appraisal of agent~$i$. 

The summation 
$\sum_{j\in V_i^+} c_{ji}(1-x_j)x_j$ can then be viewed as the measure of importance of agent~$i$ to others. Conversely, by~\eqref{eq:summationtoone}, the other term $(1 - x_i)x_i$ can be expressed as follows:   
$$
(1-x_i)x_i = \sum_{j\in V^-_i } c_{ij} (1 - x_i)x_i,  
$$
We thus interpret $(1- x_i)x_i$ as the measure of importance of others to agent~$i$. Note that by communicating with her neighbors, each agent~$i$ is able to compute these two terms by herself. 

The self-appraisal model is then designed so that agent~$i$ measures the difference of the two terms, and drives it to zero. In particular, note that if $x^*$ is an equilibrium point of system~\eqref{MODEL}, then for each agent~$i$, we have the balance equation:
$$
(1-x^*_i)x^*_i = \sum_{j\in V_i^+} c_{ji}(1-x^*_j)x^*_j. 
$$
In other words, at an equilibrium point, the importance of agent~$i$ to others is equal to the importance of others to agent~$i$.

\subsection{The main result}
{ \color{black}
In this subsection, we state the main result of the paper. We first recall that the coefficient $c_{ij}$ can be viewed as a measure of the influence of the opinion of agent~$j$ on the opinion of agent~$i$. The larger the value of $c_{ij}$ is, the more influential is agent~$j$ on agent~$i$. 
We call agent~$j$ a {\bf dominant neighbor} of agent~$i$ if $c_{ij} > 1/2$. Note that in this case, because of~\eqref{eq:summationtoone}, we have 
$$
c_{ij} > \sum_{j'\in V^-_i - \{j\}} c_{ij'}; 
$$ 
in other words, the influence of agent~$j$  on agent~$i$ exceeds the influences of all the other neighbors of agent~$i$ together. Note that if an agent~$i$ has  agent~$j$ as its unique outgoing neighbor; then, by~\eqref{eq:summationtoone}, we have $c_{ij} = 1$, and hence agent~$j$ is a dominant neighbor of agent~$i$; indeed, in this case, agent~$i$ can only take opinions from agent~$j$. On the other hand, if agent~$i$ has at least two outgoing neighbors, then, the coefficients $c_{ij}$, for $j\in V^-_i$, can be set by agent~$i$ in a way so that there is no dominant neighbor of agent~$i$.  

We assume in this paper that agents in the network act cautiously when taking opinions from others so that there is no dominant neighbor in the network. Specifically, we assume that each vertex~$i$ of $G$ has at least two outgoing neighbors. Moreover, the coefficients $c_{ij}$, for $i\to j\in E$, are such that 
\begin{equation*}\label{eq:cijlessthanonehalf}
c_{ij} \le 1/2, \hspace{10pt} \forall i\to j \in E.
\end{equation*} 
We formalize below the following facts about the self-appraisal model: (i) There is a unique stable equilibrium point $x^* := (x^*_1,\ldots, x^*_n)$ of system~\eqref{MODEL} in $\Sp[V]$; (ii) Almost all trajectories $x(t)$ converge to $x^*$, and hence $x^*$ can be interpreted as the {\it steady state} of the self-appraisals of the agents in the network, independent of their initial conditions; (iii) The self-appraisal $x^*_i$ is zero if vertex~$i$ is not a root of the graph $G$.  Furthermore, we have $x^*_i < 1/2$, for all $i\in V$. 
This, in particular, implies that in the steady state, we have
$
x^*_i < \sum_{j\neq i} x^*_j 
$, i.e., there is no single agent whose self-appraisal is greater than or equal to the sum of the self-appraisals of the remaining agents. 
}

We now state the main result in precise terms. We first summarize the key conditions for system~\eqref{MODEL}, which will be assumed in the remainder of the paper:




\begin{Assumption}\label{asmp:keyasmp}
The digraph $G = (V,E)$ is simple and rooted, with $V_r$ the set of roots. Each vertex~$i$ of $G$ has at least two outgoing neighbors. 
The coefficients $c_{ij}$ of~\eqref{MODEL}, for $i\to j\in E$, satisfy the following conditions: 
$$\sum_{j\in V^-_i}c_{ij} = 1, \hspace{10pt} \forall i\in V,$$ 
and 
$$c_{ij} \le 1/2, \hspace{10pt} \forall i\to j \in E.$$\, 
\end{Assumption} 

We now have the main result, captured by the theorem below:

\begin{theorem}\label{MAIN}
Under Assumption~\ref{asmp:keyasmp}, the self-appraisal model~\eqref{MODEL} satisfies the following properties:
\begin{enumerate}
\item The unit simplex $\Sp[V]$ is a positive invariant set.
\item There are $(n+1)$ equilibrium points of system~\eqref{MODEL}. Each of the $n$ vertices~$e_i$ of $\Sp[V]$ is an unstable equilibrium point. The remaining equilibrium point $x^*$ lies in $\Sp[V_r]$, and satisfies the
following condition:
$$ 
0<x^*_i<1/2, \hspace{10pt} \forall i\in V_r.$$
Moreover, $x^*$ is  exponentially stable.
\item For any initial condition $x(0)$ other than a vertex of $\Sp[V]$, the solution $x(t)$ of system~\eqref{MODEL} converges to $x^*$. 
\end{enumerate}\,
\end{theorem}

\begin{Remark}
Note that the cascaded system~\eqref{artif} and~\eqref{MODEL} admits a 
triangular structure (i.e., \eqref{MODEL} feeds into~\eqref{artif} but not the other way around), and hence the convergence of system~\eqref{MODEL} implies the convergence of the consensus process~\eqref{artif}. 
\end{Remark}

{ \color{black}
We  note here that the self-appraisal model~\eqref{MODEL} is similar, in its format, to the replicator dynamics (see, for example,~\cite{zhu2010heterogeneous,bomze1995lotka}),  which is also defined over the unit simplex. Yet, these are two different types of dynamical systems. Indeed, a replicator dynamics equation is given by
\begin{equation}\label{eq:replicatordynamics}
\dot x_i = x_i\(g_i(x) -   \sum^n_{j \in V}x_jg_j(x_j)\), \hspace{10pt} \forall\, i\in V.
\end{equation}
First, note that the dynamics of each $x_i$ in~\eqref{eq:replicatordynamics} depends on a global information of $x$. Yet, in the self-appraisal model~\eqref{MODEL}, the dynamics of $x_i$ depends only on $x_j$ for $j\in V^+_i$.  
Second, note that for the replicator dynamics, each subset $\Sp[V']$, for $V'\subset V$,  is a positive invariant set; indeed, from~\eqref{eq:replicatordynamics}, if $x_i = 0$, then $\dot x_i = 0$. Yet, such an invariance property does not hold in the self-appraisal model; indeed, from Theorem~\ref{MAIN}, if $\Sp[V']$ is positive invariant, then, either $V' = \{e_i\}$ for some $i\in V$, or $V'$ contains the root set $V_r$ (we prove this fact formally in the next section).   We further note that if $G$ is a strongly connected graph, then $V_r = V$, and hence from Theorem~\ref{MAIN}, there is a unique stable equilibrium point of~\eqref{MODEL} in the interior of the simplex. For the replicator dynamics, a stable equilibrium point may or may not lie in the interior of the simplex (see, for example,~\cite{bomze1995lotka}). All these facts point to the intrinsic differences between the two types of models.   

On the other hand, the self-appraisal model~\eqref{MODEL} can be viewed as a prototype of a general class of nonlinear dynamical systems defined over the unit simplex:
\begin{equation*}\label{eq:ageneralclassofsystems}
\dot x_i = - g_i(x_i) + \sum_{j\in V^+_i} c_{ji} g_j(x_j). 
\end{equation*}
with $g_i(x) \ge 0$ and $g_i(0) = g_i(1) = 0$ for all $i\in V$. Hence, 
the analysis of system~\eqref{MODEL} provided in the paper, for locating the set of equilibrium points and for establishing the global convergence, might be of independent interest. 
}

The remainder of the paper is now devoted to  establishing properties of system~\eqref{MODEL} that are needed to prove Theorem~\ref{MAIN}. 
In section~III, we focus on some basic properties of system~\eqref{MODEL}. In particular, we construct a family of subsets of the unit-simplex each of which is a positive invariant set for system~\eqref{MODEL}. In section~IV, we focus on the system behavior around a vertex of the simplex. We show that there is a closed neighborhood around each vertex in the simplex  such that if the initial condition is not the vertex, then the solution of system~\eqref{MODEL} will be away from that neighborhood after a finite amount of time. In section~V, we establish the global convergence of the self-appraisals towards the unique stable equilibrium point. 

\section{Basic Properties of the Self-Appraisal Model}

In this section, we compute time-derivatives along trajectories of system~\eqref{MODEL} at boundary points of the unit simplex. We then exhibit a family of subsets of $\Sp[V]$ (including the simplex itself), each of which is positive invariant for system~\eqref{MODEL}.

\subsection{Time-derivatives at boundary points}\label{ssec:tdbp}
In this subsection,  we assume that $x_i(t) = 0$, for some $i \in V$, and evaluate the time-derivatives $d^k x_i(t) / dt^k$, for $k \ge 1$. We first  establish in this subsection the following fact: let $d^k x_i(t) / dt^k$ be the first {\it non-vanishing} time derivative, i.e., $d^k x_i(t) / dt^k \neq 0$ and $d^l x_i(t) / dt^l = 0$ for all $l < k$; then, we have $d^k x_i(t) / dt^k > 0$. Toward that end, 
we first introduce some definitions about the graph~$G$: 

\begin{Definition}[Supporting set]
For a vertex~$i\in V$, a set  $S_i\subseteq V$ is a {\bf supporting set}, if a vertex~$j$ is in $S_i$ whenever there is a path from $j$ to $i$. The vertex~$i$ is in $S_i$ by default. 
\end{Definition}

Note that if $i$ is  a root of $G$, then $S_i = V$. Indeed, if $j\in S_i$, then so is any vertex in $V^+_j$.

Let $\gamma$ be a path in $G$, and $l(\gamma)$ denote its length. For each non-negative integer $k$, we define $S_i(k)\subseteq S_i$ as follows: if $j\in S_i(k)$, then there is a path $\gamma$ from $j$ to $i$ with $l(\gamma) \le k$. We assume that $i\in S_i(k)$ for all $k\ge 0$, and $S_i(0) = \{i\}$. By definition, we have $S_i(k-1) \subseteq S_i(k)$ for all $k\ge 1$.   Denote by $D_i(k)$ the difference between $S_i(k)$ and $S_i(k-1)$, i.e., 
\begin{equation}\label{eq:Ti(k)}
D_i(k) := S_i(k) - S_i(k-1); 
\end{equation} 
for $k = 0$, we set $D_i(0) := \{i\}$.  Note that in the case $k=1$, we have 
$S_i(1) = \{i\}\cup V^+_i$, and hence $D_i(1) = V^+_i$. 
We further note the following fact:


\begin{lem}\label{prolem25}
Let $j\in D_i(k)$ for some $k\ge 1$.  Then, $V^-_j$ intersects $S_i(k-1)$ and their intersection lies in 
$
D_i(k-1)
$.     
\end{lem}

\begin{proof}
First note that if $j\in D_i(k)$ and $\gamma$ is a path from $j$ to $i$ of shortest length, then $l(\gamma) = k$. 
Since $j\in D_i(k)$ for $k\ge 1$, there exists at least one vertex~$j'\in V^-_j$ such that $j'\in S_i(k-1)$, because otherwise the length of any path from $j$ to $i$ is greater than $k$. Now let $j'\in V^-_j\cap S_i(k-1)$, and  $\gamma'$ be a path from $j'$ to $i$; it suffices to show that $l(\gamma') \ge (k-1)$. We prove by contradiction. Suppose that $l(\gamma')< (k-1)$;  then by concatenating $j\to j'$ with $\gamma'$, we get a path $\gamma$ from $j$ to  $i$ with  $l(\gamma)< k$, which contradicts  the fact that $j\in D_i(k)$. This completes the proof.   
\end{proof}

Let $j\in D_i(k)$, and denote by $\Gamma_{ji}(k)$ the collection of paths from $j$ to $i$ of length $k$.  Choose any path $\gamma\in\Gamma_{ji}(k)$, and write
 \begin{equation*}\label{expreg}
 \gamma = j_k \to j_{k-1} \to\ldots\to j_1 \to j_0
\end{equation*}
with $j_k = j$ and $j_0 = i$. Then, each $j_l$ lies in the intersection of $S_{i}(l)$ and $V^-_{j_{l+1}}$. Thus, by Lemma~\ref{prolem25}, we have 
$$
j_l\in D_i(l), \hspace{10pt}  \forall l = 0,\ldots,k.
$$ 
For each $\gamma\in \Gamma_{ji}(k)$, define a positive number as follows:
\begin{equation*}
\alpha_{\gamma}:= \Pi^{k}_{l=1} c_{j_l j_{l-1}}.  
\end{equation*}
We further define  
\begin{equation}\label{coeff}
\alpha_{ji} := \sum_{\gamma\in\Gamma_{ji}(k)}\alpha_{\gamma}. 
\end{equation}
Note that $\alpha_{ji}$ can be defined recursively as follows: For the base case, we assume that $j\in D_i(1) =  V^+_i$; then, $\Gamma_{ji}(1)$ is a singleton, comprised only of the edge $j\to i$, and hence $\alpha_{ji} = c_{ji}$. For the inductive step, we assume that $k >1$ and $\alpha_{j'i}$, for $j'\in D_i(k - 1)$, are well defined. Let~$j\in D_i(k)$, then  
\begin{equation}\label{eq:recursivedefinitionofalphaji}
\alpha_{ji} = \sum_{j'} c_{jj'} \alpha_{j'i}
\end{equation}
where the summation is over $V^-_j \cap D_i(k-1)$, which is nonempty by Lemma~\ref{prolem25}.  
  
To state the main result of this subsection, we further need the following definition:   

\begin{Definition}[Supporting path]
Let $x\in \Sp[V]$ with $x_i = 0$ for some $i\in V$. A path 
$$
\gamma := j_k \to \ldots \to j_1 \to i    
$$
is a {\bf supporting path} for $i$ at $x$ if
\begin{equation*}\label{LENG}
\left\{
\begin{array}{ll}
x_{j_k} >0 &  \\
x_{j_l}  = 0 & \forall l < k. 
\end{array}
\right.
\end{equation*}
If, in addition,  $k$ is the least integer among the lengths of all supporting paths for $i$ at $x$, then $\gamma$ is a {\bf critical supporting path}. 
\end{Definition}
 

It should be clear that if there exists a supporting path for~$i$ at $x$, then there will be a critical supporting path. Note that all critical supporting paths for $i$ at $x$ have the same length. We have, however, also the following fact:

\begin{lem}\label{lempro23}
Let $x\in\Sp[V]$ with $x_i = 0$. If there is no supporting path for $i$ at $x$, then $x_j = 0$ for all $j\in S_i$. 
\end{lem}

\begin{proof} We prove the result by contradiction. Suppose that there exists some $j\in S_i$ such that $x_j \neq 0$. Then, by the construction of the set $S_i$, we can choose a path $\gamma$ from $j$ to $i$:
\begin{equation*}
\gamma = j \to j_{k-1} \to \ldots \to j_1 \to i.
\end{equation*} 
By truncating the path if necessary, we can assume that $x_{j_{l}} = 0$  for all $l=1,\ldots, k-1$. Then, $\gamma$ is a supporting path for $i$ at $x$,  which is a contradiction.
\end{proof}

With the preliminaries above, we compute below the time derivatives of $x_i(t)$.   For convenience, let 
$$ x^{(k)}_i(t) := d^{k}x_i(t)/ dt^k.$$
We then have the following fact:

\begin{pro}\label{pro2}
Let the initial condition $x(0)$ of system~\eqref{MODEL} be in $\Sp[V]$, other than a vertex. Suppose that  $x_i(0) = 0$ for some $i\in V$. Then, the following properties hold:
\begin{itemize}
\item[1.] If there does not exist a supporting path for $i$ at $x(0)$, then,  $i\notin V_r$ and $x_i(t) = 0$ for all $t\ge 0$.
\item[2.] If there is a critical supporting path for $i$ at $x(0)$, and the length of the path is $k$, then,  
$$  x^{(l)}_i(0) = 0   $$
for all $l<k$, and 
$$ x^{(k)}_i(0) = \sum_{j\in D_i(k)} \alpha_{ji} (1 - x_j(0))x_j(0) > 0$$
with $D_i(k)$ defined in~\eqref{eq:Ti(k)} and $\alpha_{ji}$ defined in~\eqref{coeff}.
\end{itemize}
\end{pro}

We refer to the Appendix for a proof of Proposition~\ref{pro2}.

\subsection{Positive invariant sets}

Now,  
for each subset $V'\subseteq V$, we define 
\begin{equation}\label{Eq:svunionsi}
S_{V'} := \cup_{i\in V'} S_i. 
\end{equation} 
If $V'$ is empty, then so is $S_{V'}$. Note that if $V'$ contains a root of $G$, then $S_{V'} = V$.  Denote by $S^c_{V'}$  the complement of $S_{V'}$ in $V$, i.e., 
$$
S^c_{V'}:= V - S_{V'}.
$$
Then, the set $\Sp[S^c_{V'}]$ can be described by
$$
\Sp[S^c_{V'}] = \{x\in \Sp[V]\mid x_i = 0, \forall i\in S_{V'}\}.
$$
We next establish the following fact: 

\begin{pro}\label{POSINV}
Let $V'$ be any subset of $V$. Then, $\Sp[S^c_{V'}]$ is a positive invariant set for system~\eqref{MODEL}. 
\end{pro}

\begin{proof}
First, note that if $x(0)\in \Sp[S^c_{V'}]$,  then there is no supporting path for $i$ at $x(0)$ for all $i\in S_{V'}$.  Then, by Proposition~\ref{pro2}, we have
\begin{equation*}
x_i(t) = 0, \hspace{10pt} \forall t \ge 0 \mbox{ and } \forall i\in S_{V'}.
\end{equation*}
Next, we show that $\sum^n_{i=1}x_i$ is invariant along the evolution. To see this, note that
\begin{equation*}
\begin{array}{lll}
\sum^n_{i=1} (1 - x_i) x_i & = & \sum^n_{i=1} \sum_{j\in V^-_i} c_{ij} (1 - x_i) x_i\\
& = & \sum^n_{i=1}\sum_{j\in V^+_i} c_{ji} (1 - x_j) x_j. 
\end{array}
\end{equation*}
Hence, we have 
\begin{equation*}
\sum^n_{i=1} \dot{x}_i = \sum^n_{i=1}\left(\sum_{j\in V^+_i} c_{ji} (1 - x_j) x_j  - (1-x_i)x_i\right ) = 0.
\end{equation*}
Thus, $\sum^n_{i=1}x_i(t) = 1$ for $t \ge 0$. 
It now suffices to show that
\begin{equation*}\label{Eq:posinv}
x_i(t) \ge 0, \hspace{10pt} \forall t\ge 0 \hspace{5pt}\mbox{ and } \hspace{5pt}\forall i\in V.
\end{equation*}
We prove this by contradiction. Suppose that, to the contrary, there exists a vertex~$i\in V$, an instant~$t>0$ and an $\epsilon > 0$ such that $x_i(t) = 0$ and
\begin{equation}\label{eq:outofsimplex}
x_i(t') < 0, \hspace{10pt} \forall\, t' \in (t, t+\epsilon).
\end{equation}
There are two cases: 

{\it Case I}.  
Suppose that there does not exist a supporting path for~$i$ at $x(t)$; then, by Proposition~\ref{pro2}, $x_i(t') = 0$ for all $t' \ge t$, which contradicts~\eqref{eq:outofsimplex}. 

{\it Case II}.
Suppose that there is a supporting path for~$i$ at $x(t)$; then, appealing again to Proposition~\ref{pro2}, we have that there exists a~$k$ such that $x^{(k)}_i(t) > 0$ and $x^{(l)}_i(t) = 0$
for all $l < k$. This, in particular, implies that there is an $\epsilon' > 0$ such that $x_i(t') > 0 $ for all $t' \in (t, t+ \epsilon')$, which contradicts~\eqref{eq:outofsimplex}. This completes the proof.  
\end{proof}

We state below some implications of Proposition~\ref{POSINV}. First, let $V' $ be an empty set; then $S^c_{V'} = V$. Thus, we have the following fact as an immediate consequence of Proposition~\ref{POSINV}: 
\begin{cor}\label{spvvv}
The unit simplex $\Sp[V]$ is a positive invariant set for system~\eqref{MODEL}. 
\end{cor}

Recall that $V_r$ is the root set of $G$; we establish the following fact as another corollary to Proposition~\ref{POSINV}:
 
\begin{cor}\label{spvr}
The subset $\Sp[V_r]$ of $\Sp[V]$ is a positive invariant set for system~\eqref{MODEL}.
\end{cor}

\begin{proof}
Let $V' := V - V_r$; we prove the result by  showing that  $S^c_{V'} = V_r$, and then, appealing to Proposition~\ref{POSINV}. It should be clear that $S^c_{V'} \subseteq V_r$, and we prove that $S^c_{V'} \supseteq V_r$. 
The proof is carried out by contradiction. Suppose that there is a root~$i\in S_{V'}$. From~\eqref{Eq:svunionsi}, we know that there is a vertex~$j\in V'$, together with a path $\gamma$ from~$i$ to~$j$. But then, since $i$ is a root, each vertex of $\gamma$ is a root. In particular, we have that $j\in V_r$,  which contradicts the fact that $j\in V'$. This completes the proof.  
\end{proof}






\section{Vertices Are Repellers}
In this section, we focus on the system behavior around a vertex of the simplex. First, observe that each vertex $e_i\in \R^n$ of the simplex is an equilibrium point of system~\eqref{MODEL}. We show below that each $e_i$ is unstable, and in fact, it is a {\it repeller}. To this end, 
we choose an $\epsilon\in (0,1)$, and define $P_i(\epsilon)\subset \Sp[V]$ as follows:
\begin{equation*}\label{Vi}
P_i(\epsilon) := \big \{x\in \Sp[V]\mid x_i \ge \epsilon \big \}.
\end{equation*}  
It should be clear that each $P_i(\epsilon)$ is a closed neighborhood of $e_i$ in $\Sp[V]$, and 
$$ P_i(\epsilon) \supsetneq P_i(\epsilon'), \hspace{10pt}\mbox{if } \epsilon < \epsilon'. $$
Denote by $f(x)$ the vector field of system~\eqref{MODEL}, and by $f_i(x)$ the $i$-th component of $f(x)$, i.e., 
\begin{equation}\label{eq:fi(x)}
f_i(x):= -(1-x_i)x_i + \sum_{j\in V^+_i}c_{ji} (1-x_j)x_j.
\end{equation}
We next establish the following fact:

\begin{pro}\label{VR}
 Suppose that for a vertex $i\in V$, there exists $\epsilon\in (0,1)$  such that  
$
c_{ji} \le \epsilon
$ 
for all $j\in V^+_i$.  Then, there exists an $\alpha_i\in (0,\epsilon)$ such that 
$
f_i(x) < 0 
$
for all $x\in P_i(\alpha_i) - \{e_i\}$. 

\end{pro}

\begin{Remark}
We recall that $c_{ji}$ measures the influence of agent~$i$ on agent~$j$. Proposition~\ref{VR} then says that if the influence of agent~$i$ on each of her incoming neighbors is less than $\epsilon$, then the self-appraisal $x_i(t)$ of agent~$i$ decreases along the evolution and stays below~$\epsilon$.     
\end{Remark}

\begin{proof}[Proof of Proposition~\ref{VR}]
First, note that if vertex $i$ does not have any incoming neighbor, then
$$
f_i(x) = -x_i(1- x_i), 
$$
and hence, $f_i(x) < 0$ as long as $x_i \in (0,1)$. Hence, in this case, Proposition~\ref{VR} holds for any $\alpha_i\in (0,\epsilon)$. We thus assume that $V^+_i$ is nonempty for the rest of the proof. 

To prove the result, we first apply the condition that $c_{ji}\le \epsilon$, and obtain
\begin{equation}\label{f1}
\begin{array}{lll}
f_i(x)  \le  - (1 - x_i)x_i + \epsilon \sum_{j\in V^+_i}(1-x_j)x_j.
\end{array}
\end{equation}
Now, fix the value of $x_i$, and consider the following optimization problem:
\begin{equation*}
\begin{array}{l}
\max \sum_{j\in V^+_i} (1 - x_j)x_j \vspace{3pt}\\
\text{s.t. } \sum_{j\in V^+_i} x_j \le 1 - x_i  \hspace{5pt} 
\text{ and }  \hspace{5pt} x_j \ge 0.
\end{array}
\end{equation*} 
Since the function $\sum_{j\in V^+_i} (1 - x_j)x_j $ is strictly concave, the maximum is achieved uniquely at 
\begin{equation*}
x_{j} = (1 - x_i)/d_i, \hspace{10pt} \forall j\in V^+_i
\end{equation*}
where $d_i$ is the cardinality of $V^+_i$. 
Thus,  
\begin{equation}\label{opt}
\max \sum_{j\in V^+_i} (1 - x_j)x_j = (1 - x_i)(1 - (1 - x_i)/d_i).
\end{equation} 
We then combine~\eqref{f1} and~\eqref{opt}, and get
\begin{equation*}
f_i(x) \le -(1 - {\epsilon}/{d_i} )(1 - x_i) \left(x_i - \epsilon \,\frac{d_i - 1}{d_i - \epsilon}\right).
\end{equation*}
For the first term on the right hand side of the expression, we have
$
(1 - \epsilon/d_i)>0
$. For the second, we have that if $x\neq e_i$, then $(1 - x_i) > 0$. For the last term, we let
\begin{equation}\label{prealpha}
\alpha_i:= \epsilon -  \frac{1}{2} \frac{\epsilon(1 - \epsilon)}{ d_i - \epsilon}.
\end{equation}
Then, we have
\begin{equation*}
x_i - \epsilon\, \frac{d_i - 1}{d_i - \epsilon}   
\ge \frac{1}{2}\frac{\epsilon(1 - \epsilon)}{d_i - \epsilon} >0, \hspace{10pt} \forall x_i \ge \alpha_i.
\end{equation*}
Combining these facts, we arrive at the result that   
$
f_i(x) < 0$, for all  $x\in P_i(\alpha_i) - \{e_i\} 
$. 
\end{proof}

Recall that $c_{ij}\le 1/2$ for all $i\to j\in E$. Thus, if we let $\epsilon = 1/2$, then following~\eqref{prealpha}, we have
\begin{equation*}
\alpha_i= \frac{1}{2} - \frac{1}{4(2d_i - 1)} \le \frac{1}{2} - \frac{1}{4(2n - 3)}.
\end{equation*}
The equality holds if and only if $V^+_i = V - \{i\}$, and hence $d_i = n-1$.  
For convenience, let $\alpha$ be defined as follows: 
\begin{equation}\label{alpha}
\alpha := \frac{1}{2}  - \frac{1}{4(2n - 3)}.
\end{equation}
It is the sharp upper-bound for all $\alpha_i$'s.  This number $\alpha$ will be fixed for the remainder of the paper. 

Following Proposition~\ref{VR},  we now describe some relevant properties of system~\eqref{MODEL}.  First, we define a subset $Q\subset \Sp[V]$ as follows:
\begin{equation}\label{eqDefQ}
Q: =\{x\in \Sp[V]\mid x_i\le \alpha,\, \forall i\in V\}.
\end{equation}
Note that $Q$ can be expressed as follows
\begin{equation}\label{reexDefq}
Q = \bigcap_{i\in V} \cl \big(\Sp[V] - P_i(\alpha)\big)
\end{equation}
where $ \operatorname{cl}(\cdot)$ denotes the closure of a set in $\bb{R}^n$. An illustration of $Q$ in the case $n = 3$ is given in Figure~\ref{Pep}.

\begin{figure}[h]
\begin{center}
\includegraphics[scale=.3]{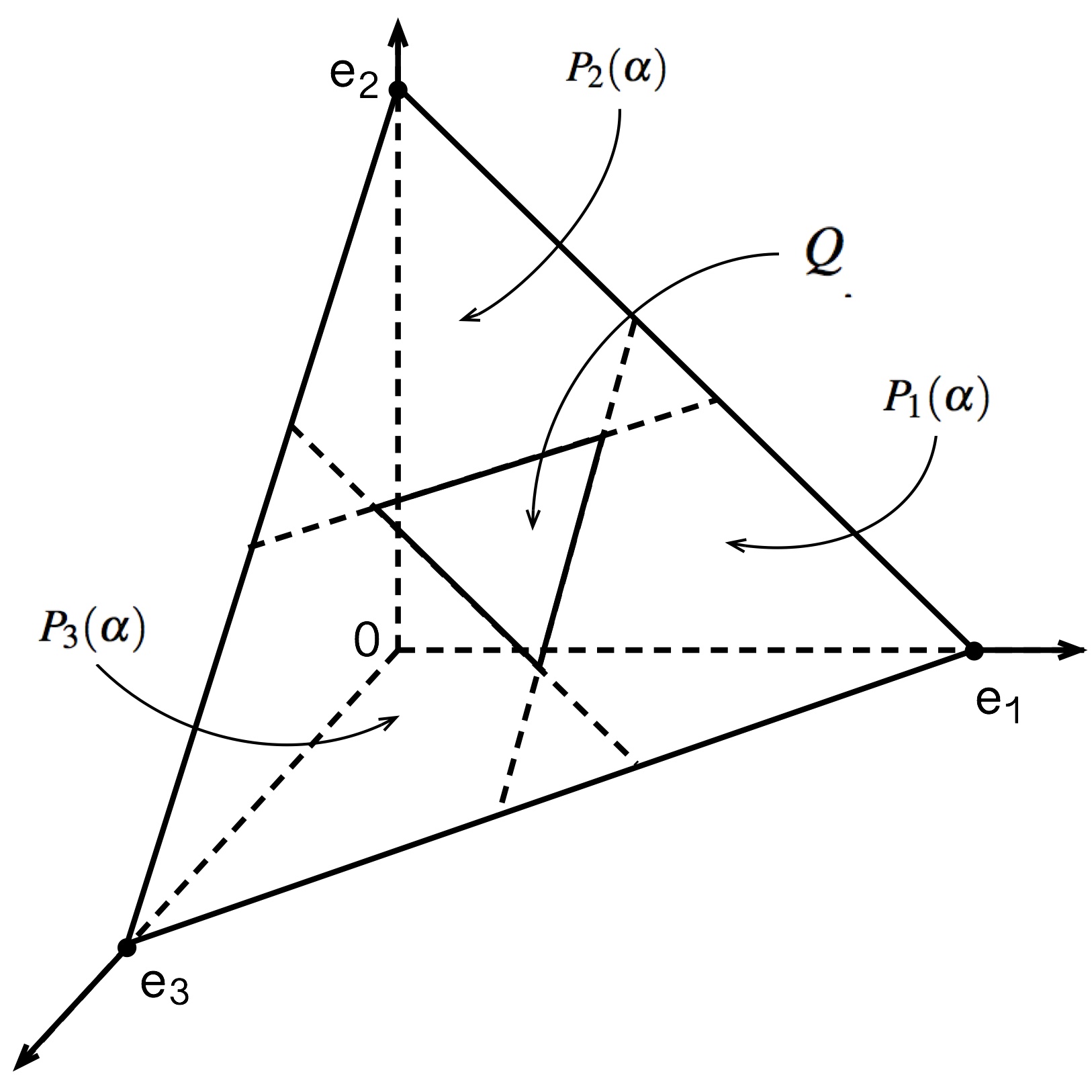}
\caption{Illustration of how $P_i(\alpha)$'s and $Q$ are defined and how these sets are distributed over the simplex in the case $n=3$.}
\label{Pep}
\end{center}
\end{figure}

We obtain some results as corollaries to Proposition~\ref{VR}:

\begin{cor}\label{R1}
If $x$ is a non-vertex equilibrium point of system~\eqref{MODEL}, then $x\in Q$.  
\end{cor} 

\begin{proof} 
If $x\in \Sp[V] - Q$, then $x\in P_i(\alpha)$ for some $i\in V$. Since $x$ is not a vertex, by Proposition~\ref{VR}, we have $f_i(x) < 0$. Thus, $x$ can not be an equilibrium point if $x\notin Q$. 
\end{proof}

We now show that the set $Q$ is also positive invariant:    

\begin{cor}\label{DefQ}
The set $Q$ is positive invariant for system~\eqref{MODEL}.  For any initial condition $x(0)\in \Sp[V]$ other than a vertex, the trajectory $x(t)$ enters  $Q$ in a finite amount of time. 
\end{cor}

\begin{proof}
We first show that $Q$ is positive invariant. From~\eqref{reexDefq}, 
it suffices to show that each  $\cl(\Sp[V] - P_i(\alpha))$ is positive invariant. By construction, 
$x_i \le  \alpha $ if and only if $x\in \cl(\Sp[V] - P_i(\alpha))$. On the other hand, by Proposition~\ref{VR},  $f_i(x) < 0$ for all $x\in P_i(\alpha)$. It then follows that $\cl(\Sp[V] - P_i(\alpha))$ is positive invariant.

We now show that if $x(0)\neq e_i$ for any $i\in V$, then there exists an instant $T\ge 0$ such that $x(t)\in Q$ for all $t\ge T$. Similarly, it suffices to show that for each $i\in V$, there exists an instant $T_i\ge 0$ such that $x(t) \in \cl(\Sp[V] - P_i(\alpha))$ for all $t\ge T_i$. First, note that if $x_i(0)\leq \alpha$, then $x\in \cl(\Sp[V] -P_i(\alpha))$. By the fact that $ \cl(\Sp[V] -P_i(\alpha))$ is positive invariant, we can choose $T_i = 0$. We now assume that $x_i(0) > \alpha$. Then, 
$$
P_i(\alpha) \supsetneq P_i(x_i(0)), 
$$
and hence 
the set  $P_i(\alpha) - P_i(x_i(0))$ is nonempty. 
Let 
\begin{equation*}
v_i := \inf\big \{ |f_i(x) | \mid x\in P_i(\alpha) - P_i(x_i(0))\big \}
\end{equation*}  
Since $\cl\left(P_i(\alpha) - P_i(x_i(0))\right)$ is compact over which $f_i(x)$ is strictly negative, we have $v_i>0$. 
Let 
$$
T_i := \frac{x_i(0) - \alpha} { v_i} + 1
$$
Then, by construction, the trajectory $x(t)$ will enter  $\cl(\Sp[V] - P_i(\alpha))$ in no more than 
$T_i$ units of time. This completes the proof. 
\end{proof}


\section{The Global Convergence of \\The Self-Appraisal Model}
In this section, we establish the global convergence of system~\eqref{MODEL}. The section is divided into four parts: In subsection~\ref{ssec:nonroot}, we show that the self-appraisals of non-root agents converge to zero. In subsection~\ref{ssec:VB}, we establish the second part of Theorem~\ref{MAIN}, i.e., there is only one stable equilibrium point $x^*$ of system~\eqref{MODEL}, which lies in $\Sp[V_r]$, and moreover, $0< x^*_i< 1/2$, for all $i\in V_r$. In subsection~\ref{ssec:VC}, we show that if the initial condition $x(0)$ of system~\eqref{MODEL} is such that $x(0)\in Q$, and $x_i(0) = 0$ for $i\notin V_r$, then, the trajectory $x(t)$ converges to the unique stable equilibrium point~$x^*$. In subsection~\ref{ssec:VD}, we combine all the results and prove Theorem~\ref{MAIN}.

\subsection{Convergence of self-appraisals of non-root agents}\label{ssec:nonroot}
An agent~$i$ is said to be a {\bf non-root} agent if~$i$ is not a root of $G$. We establish in this subsection the result that the self-appraisals of non-root agents decay to zero. 

\begin{pro}\label{pro31}
Let $x(0) \in \Sp[V]$ be other than a vertex. Then, $x_i(t)$ converges to $0$ for all $i\notin V_r$.  
\end{pro}


To prove Proposition~\ref{pro31}, we first etablish the following fact:

\begin{lem}\label{pro31lem1}
If $x_i(t)$ converges to $0$, then so does $x_j(t)$ for all $j\in V^+_i$. 
\end{lem}

\begin{proof}
We prove Lemma~\ref{pro31lem1} by contradiction. Suppose that there is a vertex~$j\in V^+_i$ such that $x_j(t)$ does not converge to $0$. Then, there is an $\epsilon\in (0,1/2)$, and an infinite time sequence $\{t_k\}_{k\in\bb{N}}$ with $\lim_{k\to \infty} t_k = \infty$ such that $x_j(t_{k})>\epsilon$ for all $k\in \mathbb{N}$. 

Since there is an upper bound for $|f_j(x)|$ for all $x\in\Sp[V]$, there is a finite duration time $\tau>0$ such that 
$$
x_j(t_k + \tau') \ge \epsilon/2, \hspace{10pt} \forall \tau'\in [0,\tau] \mbox{ and } \forall k\in\bb{N}.
$$
Rescale $\tau$, if necessary, so that $\tau\le 1$.  On the other hand, 
by Corollary~\ref{DefQ}, there exists an instant $T\ge 0$ such that $x(t)\in Q$ for all $t\ge T$. This, in particular, implies that  $x_j(t)\le 1/2$ for all $t\ge T$. Without loss of generality, assume that $t_k\ge T$ for all $k\in \bb{N}$, and hence 
$$
x_j(t)\in [\epsilon/2, 1/2]
$$ 
for all $t\in [t_k,t_k + \tau]$ and for all $ k\in\bb{N}$. 
It then follows that 
\begin{equation}\label{cjixjt}
c_{ji}(1 - x_j(t))x_j(t) \ge c_{ji}(1 - \epsilon/2)\epsilon/2.
\end{equation}
for all $t\in [t_k,t_k + \tau]$ and for all $ k\in\bb{N}$.

For convenience, let 
$
\delta:=  c_{ji}(1 - \epsilon/2)\epsilon/2 
$.  Since $x_i(t)$ converges to $0$, there is an instant $T'$  such that 
\begin{equation}\label{xitcon}
x_i(t) \le \tau\delta/2 , \hspace{10pt} \forall t\ge T'. 
\end{equation}
Using the fact that $\tau <1$ and $(1 - x_i(t)) < 1$, we obtain
\begin{equation}\label{xitxit}
(1 - x_i(t) )x_i(t) < \delta/2, \hspace{10pt} \forall t\ge T'.
\end{equation} 
We recall that $f_i(x)$ (defined in~\eqref{eq:fi(x)}) is given by
$$
f_i(x) = -(1 - x_i)x_i + \sum_{j\in V_i^+}c_{ji} (1 - x_j) x_j.
$$
Combining~\eqref{cjixjt} and~\eqref{xitxit}, we then obtain
\begin{equation*}
f_i(x(t)) > \delta/2, \hspace{10pt} \forall t\in [t_k, t_k + \tau] 
\end{equation*}
as long as $t_k \ge T'$. It then follows that 
\begin{equation*}
x_i(t_k + \tau) = x_i(t_k) + \displaystyle \int^{t_k + \tau}_{t_k} f_i(x(t)) dt > \delta \tau /2
\end{equation*} 
which contradicts~\eqref{xitcon}. This completes the proof.  
\end{proof}

Recall that $S_i\subset V$ is the supporting set of~$i$, and $S_i(k)\subset V$ is defined such that if $j\in S_i(k)$, then there is a path $\gamma$ from $j$ to $i$, with $l(\gamma)\le k$. 
Also, recall that for a subset $V'$ of $V$, we have defined $S_{V'} =  \cup_{i\in V'} S_i$. In the remainder of this subsection, we focus on the case where $V'$ is $V_r$, the set of roots of $G$.

First, note that by the definition of $V_r$, we have $S_{V_r} = V$; indeed, for any vertex~$j$, there is a path from~$j$ to a root of $G$. 
Now, for each $k \ge 0$,  we define
$$
S_{V_r}(k) := \cup_{i\in V_r} S_i(k).    
$$
For $k = 0$, we have $S_{V_r}(0) = V_r$. 
It should be clear that
$
S_{V_r}(k-1) \subseteq S_{V_r}(k)
$, and moreover, there is an integer $m$ such that 
\begin{equation}\label{ss}
S_{V_r}(m) = V.   
\end{equation} 
We now assume that $m$ is the least integer such that~\eqref{ss} holds. 
Let 
\begin{equation*}
D_{V_r}(k) := S_{V_r}(k)  - S_{V_r}(k-1).  
\end{equation*}
With Lemma~\ref{pro31lem1} and the notations above, we prove Proposition~\ref{pro31}.

\begin{proof}[Proof of Proposition~\ref{pro31}]
We prove the proposition by subsequently showing that 
$$
\lim_{t\to\infty} x_i(t) = 0, \hspace{10pt} \forall i\in D_{V_r}(k), 
$$
for $k=1,\ldots, m$.  

{\it Base case}. We prove for $k = 1$. Let $\phi_r$ be a function on $\Sp[V]$ defined as follows:
\begin{equation}\label{eq:phirx}
\phi_r(x):= \sum_{i\in V_r}x_i. 
\end{equation}
We first compute the time derivative of $\phi_r(x)$ along a trajectory $x(t)$ of system~\eqref{MODEL}: first, let 
$E'$ be a subset of $E$ defined by
$$
E': = \{i\to j\in E \mid  i\in D_{V_r}(1), \, j\in V_r\};
$$
we then have
\begin{equation}\label{eq:phirinc}
\frac{d}{dt}{\phi}_r(x) 
  = \sum_{i\to j\in E'}  c_{ij}(1 - x_i )x_i \ge 0. 
\end{equation}
Note that by its definition~\eqref{eq:phirx}, we have that $\phi_r(x)\le 1$ for all $x\in \Sp[V]$. Also, note that from~\eqref{eq:phirinc}, $\phi_r(x(t))$ is monotonically increasing in~$t$. So,     
$
\lim_{t\to\infty}\phi_r(x(t))
$ exists. 

Next, we show that  $\dot{\phi}_r(x(t))$ converges to $0$. The proof will be similar to the proof of Lemma~\ref{pro31lem1}, and is carried out  by contradiction. Suppose that there exists an $\epsilon$, and a time sequence $\{t_k\}_{k\in \bb{N}}$ with $\lim_{k\to\infty}t_k = \infty $, such that 
$$
\dot{\phi}_r(x(t_k)) \ge \epsilon, \hspace{10pt} \forall k\in \bb{N}.
$$ 
Note that $|\ddot{\phi}_r(x(t))|$ is bounded above for any $x(t) \in \Sp[V]$, and hence there is a duration time $\tau > 0$ such that 
$$
\dot{\phi}_r(x(t)) \ge \epsilon/2, \hspace{10pt} \forall t\in [t_k, t_k + \tau] \hspace{5pt}\mbox{ and } \hspace{5pt} \forall k\in \bb{N}. 
$$
By passing to a subsequence, we can assume that 
$t_k + \tau < t_{k+1}$ for all $k\in \bb{N}$. 
Thus, we have
$$
\lim_{t\to \infty}\phi_r(x(t)) - \phi_r(x(0)) \ge \sum_{k\in \bb{N}} \int^{t_k + \tau}_{t_k} \dot{\phi}_r(x(s))ds.  
$$
By construction, each summand on the right hand side is  bounded below by $\epsilon\tau/2$, and hence the summation is infinite which contradicts the fact that $\lim_{t\to \infty}\phi_r(x(t))$ exists. 
 
Since $\dot{\phi}_r(x(t))$ converges to $0$  and each summand of $\dot{\phi}_r(x(t))$ in~\eqref{eq:phirinc} is non-negative, we have
 $$
\lim_{t\to \infty}x_i(t)(1- x_i(t)) = 0, \hspace{10pt} \forall i\in D_{V_r}(1).  
$$
Thus, $x_i(t)$ converges to either $0$ or $1$. 
Since $x(0)$ is not a vertex, by Corollary~\ref{DefQ}, the trajectory $x(t)$ enters $Q$ in a finite amount of time.  We thus conclude that     
 $x_i(t)$ converges to $0$ for all $i\in D_{V_r}(1)$.  

{\it Inductive step}. We assume that $x_j(t)$ converges to $0$ for all $j\in D_{V_r}(k - 1)$ for $k \ge 2$, and prove that $x_i(t)$ converges to $0$ for all $i\in D_{V_r}(k)$. Fix a vertex~$i$ in $D_{V_r}(k)$; then, from Lemma~\ref{prolem25}, there is a vertex~$j\in V^-_i \cap D_{V_r}(k-1)$. From the induction hypothesis, we have that $x_j(t)$ converges to $0$. Since $i\in V^+_j$,  by Lemma~\ref{pro31lem1}, $x_i(t)$ also converges to $0$. This completes the proof.
\end{proof}

\subsection{The non-vertex equilibrium point}\label{ssec:VB}
In this subsection, we prove that there exists a unique non-vertex equilibrium point of system~\eqref{MODEL}.  First, recall that $Q$ is defined in~\eqref{eqDefQ} and $\Sp[V_r]$ is the convex hull spanned by $e_i$, for $i\in V_r$. We then define
\begin{equation}\label{DefQr}
Q_r : = Q \cap \Sp[V_r].
\end{equation}
From Corollaries~\ref{spvr} and~\ref{DefQ}, both $\Sp[V_r]$ and $Q$ are positive invariant sets for system~\eqref{MODEL}, and hence  
 $Q_r$ is also a positive invariant set. 
We establish in this subsection the following result:

\begin{pro}\label{EQUI}
There is a unique non-vertex equilibrium point $x^*$ of system~\eqref{MODEL}, which lies in $Q_r$, and $0 < x^*_i < 1/2 $ for all $i\in V_r$. Moreover, $x^*$ is exponentially stable.  
\end{pro}

To prove Proposition~\ref{EQUI}, we re-write system~\eqref{MODEL} into a matrix form. We first recall the definition of infinitesimally stochastic matrices:

\begin{Definition}[Infinitesimally stochastic matrix] 
A square matrix $A$ is an {\bf infinitesimally stochastic matrix} if each off-diagonal entry of $A$ is non-negative, and the entries of each row sum to zero.  
\end{Definition}

We now define an infinitesimally stochastic matrix $C$ by specifying its off-diagonal entries: define the $ij$-th entry of $C$ to be $c_{ij}$ if $i\to j\in E$, and $0$ otherwise. The diagonal entries of $C$ are then uniquely determined by the condition that the rows of $C$ sum to zero. We note here that the negative of the matrix $C$ is known as  the Laplacian matrix of a weighted graph $G$, with weights $c_{ij}$ for $i\to j \in E$.

For a vector $x$ in $\Sp[V]$, let $X$ be a diagonal matrix defined as follows: $$X:= \diag(x_1,\ldots,x_n).$$ 
Then, system~\eqref{MODEL} can be rewritten into the following matrix form:
\begin{equation}\label{MODEL1}
\dot{x} = C^{\top}(I - X)x.
\end{equation}
For simplicity, but without loss of generality, we assume in the rest of this subsection that 
\begin{equation}\label{eq:defVr1-k}
V_r = \{1,\ldots, k\}.
\end{equation}
Note that $k \ge 3$, i.e., the cardinality of $V_r$ is at least three; indeed, by the assumption, each vertex~$i$ of $G$ has at least two outgoing neighbors, and moreover, an outgoing neighbor of a root is also a root. 
Following~\eqref{eq:defVr1-k}, we have that $C$ is a lower block-triangular matrix:  
\begin{equation}\label{blockc}
C = 
\begin{pmatrix}
C_{11} & 0\\
C_{21} & C_{22}
\end{pmatrix}
\end{equation}
with $C_{11}$ a $k$-by-$k$ infinitesimally stochastic matrix.

We next state some known facts about the matrix $C$: (i) Since $G$ is rooted and $c_{ij}>0$ for all $i\to j\in E$,  the matrix $C$ has zero as a simple eigenvalue  while all the other eigenvalues of $C$ have negative real parts (see, for example,~\cite{ren2005consensus}). Using the fact that $C_{11}$ is an infinitesimally stochastic matrix, we have that $C_{11}$ has zero as a simple eigenvalue, with all the other eigenvalues of $C_{11}$ having negative real parts, and that $C_{22}$ is a {\it stable} matrix, i.e., all the eigenvalues of $C_{22}$ have negative real parts. (ii) 
Let $v\in \R^n$ (resp. $v' \in \R^k$) be the left-eigenvector of $C$ (resp. $C_{11}$) corresponding to the zero eigenvalue; then   
$
v = (v',0)
$.  
Scale $v$ such that $\sum^n_{i=1} v_i = 1$; then all entries of $v'$ are positive, and they sum to one (see, for example,~\cite{farina2011positive}).

We next state another relevant property of $v$. Recall that $c_{ij} \le 1/2$ for all $i\to j\in E$. It then follows that 
\begin{equation}\label{vileq}
v_i \le 1/3, \hspace{10pt} \forall i\in V.
\end{equation}
To see this, note that  $C^\top v = 0$, and hence  
$$
v_i = \sum_{j\in V^+_i} c_{ji} v_j 
$$ 
for all $i\in V$.  Thus, 
$$
v_i \le \frac{1}{2} \sum_{j\neq  i} v_j = \frac{1}{2} (1 - v_i),
$$
from which~\eqref{vileq} holds.

We now return to the proof of Proposition~\ref{EQUI}. We establish the proposition by first showing that the unique equilibrium point $x^*$ lies in $Q_r$, with $0 < x^*_i < 1/2 $ for all $i\in V_r$, and then showing that $x^*$ is exponentially stable. This is done in Lemmas~\ref{EQUlem2} and~\ref{esie} below.

\begin{lem}\label{EQUlem2}
There exists a unique non-vertex equilibrium point $x^*$ of system~\eqref{MODEL}. Moreover, $x^*\in Q_r$, and $0 < x^*_i < 1/2$ for all $i\in V_r$. 
\end{lem}

\begin{proof}
Note that from~\eqref{MODEL1}, if $x^*$ is an equilibrium point of system~\eqref{MODEL}, then 
\begin{equation*}\label{equieq}
C^{\top}(I- X^*) x^* = 0.
\end{equation*} 
Since the null space of $C^{\top}$ is spanned by the single vector $v$, we must have
\begin{equation}\label{muvv}
(I- X^*) x^* = \mu v 
\end{equation}
for some $\mu \in\R$. Note that the entries of $(1 - X^*)x^*$ and of $v$ are all nonnegative, and thus we have  $\mu \ge 0$. 

We first show that $x^*_i = 0$ for all $i> k$. 
Since $v_i = 0$ for all $i>k$, we have 
$
(1 - x^*_i)x^*_i  = 0$ for all $i > k$. 
Since $x^*$ is not a vertex, we have $x^*_i \neq 1$, and hence 
$
x^*_i = 0$ for all $i > k$. 

We now solve for $x^*_i$ for all $i\le k$. Following~\eqref{muvv}, we have 
\begin{equation*}\label{eqsolfx}
(1 - x^*_i)x^*_i = \mu v_i, \hspace{10pt} \forall i\le k. 
\end{equation*}
From Corollary~\ref{R1}, we have 
$
x^*_i < 1/2$ for all  $i\le k$. 
Thus, for a fixed $\mu$,  we can solve for $x^*_i(\mu)$ as
\begin{equation*}\label{xalpha}
x^*_i(\mu) = \frac{1}{2}\left (1 - \sqrt{1 - 4\mu v_i}\right ). 
\end{equation*}
It now suffices to show that there is a unique positive $\mu$ such that 
$
\sum^k_{i=1} x^*_i(\mu) = 1
$. 
Define
\begin{equation*}
\psi(\mu) := \sum^k_{i=1} x^*_i(\mu);  
\end{equation*}
note that $\psi(0) = 0$, and moreover, $\psi(\mu)$ is strictly monotonically increasing in~$\mu$ as long as 
$
(1 - 4\mu v_i)\ge 0
$,   for all $i=1,\ldots,k$.  
Without loss of generality, we assume that $v_1 \ge v_i$ for all $i> 1$. Let  
$
\mu_1 := 1/(4v_1)
$; 
it then suffices to show that  
$
\psi(\mu_1) >1
$.

Let $w_i:= v_i/v_1$; then,  $\psi(\mu_1)$ can be expressed as follows:
\begin{equation*}
\psi(\mu_1) = \frac{k}{2} - \frac{1}{2}\sum^k_{i=2}\sqrt{1 - w_i}.   
\end{equation*}
We then consider the 
following optimization problem:
\begin{equation*}
\begin{array}{l}
\max \sum^{k}_{i=2}\sqrt{1 - w_i}  \vspace{3pt}\\
\text{s.t. } 0 \le w_i \le 1  \hspace{5pt}
\text{ and } \hspace{5pt} \sum^{k}_{i=2} w_i  \ge  2. 
\end{array}
\end{equation*}
Note that the first constraint $0 \le w_i \le 1$ comes from the assumption that $v_i \le v_1$, and the second constraint $\sum^{k}_{i=2} w_i  \ge  2$ follows from~\eqref{vileq}; indeed, we have 
$$
\sum^{k}_{i=2} w_i  = (1 - v_1)/v_1 \ge 2. 
$$
Since the function $\sum^{k}_{i=2}\sqrt{1 - w_i}$ is strictly concave in $w_i$, the maximum is achieved uniquely at
\begin{equation*}
w_i  = \frac{2}{k-1},\hspace{10pt} \forall i=2,\ldots,k. 
\end{equation*}
We recall that the cardinality of the root set $V_r$ is at least three, i.e., $k \ge 3$, and hence $$w _i = \frac{2}{k-1} \le 1,$$ which satisfies the constraint that $0\le w_i \le 1$.  
%
With the above values of $w_i$, for $i = 2,\ldots, k$, we have
\begin{equation}\label{keyeq}
\max \sum^{k}_{i=2}\sqrt{1 - w_i}  =  (k -1) \sqrt{1 - \frac{2}{k-1}} < (k-2).
\end{equation}
The last inequality in~\eqref{keyeq} holds because
\begin{equation*}
1 - \frac{2}{k-1} < 1 - \frac{2}{k-1} + \frac{1}{(k-1)^2}  = \left (\frac{k-2}{k-1}\right)^2.
\end{equation*}
Following~\eqref{keyeq}, we then have 
\begin{equation*}
\psi(\mu_1) > \frac{k}{2} - \frac{1}{2} (k-2)  = 1. 
\end{equation*}
which completes the proof.
\end{proof}

Note that in the case when $G$ is  strongly connected, the set of equilibrium points of system~\eqref{MODEL} coincides with the set of equilibrium points of the opinion system studied in \cite{Bullo1,Bullo5}, though with a completely different dynamical system.  

It now remains to show that $x^*$ is exponentially stable. We first have some preliminaries. 

\begin{Definition}[Induced graphs] 
Let $A\in \bb{R}^{n\times n}$ be an infinitesimally stochastic matrix. We say a digraph $G_A = (V,E_A)$ of $n$ vertices is {\bf induced by $A$} if the edge set $E_A$ satisfies the following condition: $i\to j\in E_A$ if and only if the $ij$-th entry of $A$ is positive. 
\end{Definition}

It is well known that if $G_A$ is rooted, then $A$ has zero as a simple eigenvalue and all the other eigenvalues have negative real parts (see, for example, \cite{ren2005consensus,luc04}). With the notions above, we establish the following result:

\begin{lem}\label{esie}
Let $x^*$ be the non-vertex equilibrium point of system~\eqref{MODEL}. Then, $x^*$ is exponentially stable.
\end{lem}

\begin{proof}
Let $J(x^*)$ be the Jacobian matrix of the vector field $f(x)$ at $x^*$, i.e.,
\begin{equation*}
J(x^*):= \frac{\partial f(x^*)}{\partial x}. 
\end{equation*}
Then, by a direct computation, 
\begin{equation*}
J(x^*):= C^{\top}(I - 2X^*).
\end{equation*} 
Since $x^*\in Q_r$, we have
\begin{equation*}
1 - 2x^*_i \ge 1 - 2 \alpha = 1/(4n - 6), \hspace{10pt} \forall i\in V, 
\end{equation*}
with $\alpha$ defined in~\eqref{alpha}. 
In particular, each diagonal entry of the diagonal matrix $(I - 2X^*)$ is positive. Let $\mathbf{1}\in \bb{R}^n$ be a vector of all ones; then, we have  
$$
J(x^*)^\top \mathbf{1}= (I - 2X^*)C\mathbf{1} = 0,
$$
and hence $J(x^*)^\top$ is an infinitesimally stochastic matrix. Moreover, the digraph induced by $J(x^*)^{\top}$ is the same as the digraph induced by $C$, which is $G$. Since $G$ is rooted, $J(x^*)$ has zero as a simple eigenvalue and all the other eigenvalues of $J(x^*)$ have negative real parts. Let $L$ be a linear subspace of $\R^n$ perpendicular to the vector $\mathbf{1}$: 
\begin{equation*}
L := \big \{v\in \mathbb{R}^n \mid \langle v,\mathbf{1}\rangle = 0\big\}.
\end{equation*} 
Then,  $L$ can be viewed as the tangent space of $\Sp[V]$ at $x$ for all $x\in \Sp[V]$. Note that $L$ is invariant under $J(x^*)$, i.e., for any vector $v\in L$, we have 
$
J(x^*) v\in L
$; indeed, we have 
$$
\langle J(x^*) v, \mathbf{1} \rangle  = \langle v, J(x^*)^\top \mathbf{1} \rangle = 0.
$$
This, in particular, implies that the eigenvalues of $J(x^*)$ have negative real parts when restricted to $L$, and hence the equilibrium point $x^*$ of system~\eqref{MODEL} is exponentially stable.  
\end{proof}

Proposition~\ref{EQUI} is then established by combining Lemmas~\ref{EQUlem2} and~\ref{esie}.

\subsection{System convergence over $Q_r$}\label{ssec:VC} 
In this subsection, we establish the convergence for a special class of trajectories $x(t)$ of system~\eqref{MODEL}. These are the trajectories whose initial conditions are in the positive invariant set $Q_r$ defined in~\eqref{DefQr}. Precisely, we have the following:


\begin{pro}\label{convQr}
Suppose that the initial condition $x(0)$ of system~\eqref{MODEL} is in $Q_r$; then, the trajectory $x(t)$ converges to the non-vertex equilibrium point $x^*$. 
\end{pro}

We establish Proposition~\ref{convQr} below. First, recall that we have labeled the vertices of $G$ such that the root set $V_r$ is comprised of the first $k$ vertices. Let 
\begin{equation*}
x' := (x_1,\ldots.x_k)
\end{equation*} 
and similarly, define $$X' := \diag(x_1,\ldots,x_k).$$ Note that if $x\in Q_r$, then $x = (x',0)$.  Since $Q_r$ is positive invariant,  the dynamics of $x'$ is simply given by  
\begin{equation}\label{subsys}
\dot{x}'  = C^{\top}_{11} (I - X') x'
\end{equation}  
with $C_{11}$ the $k$-by-$k$ block matrix defined in~\eqref{blockc}. 
To prove Proposition~\ref{convQr}, it suffices to show that for any initial condition $x'(0)$, the solution $x'(t)$ of system~\eqref{subsys} converges to $x'^* := (x^*_1,\ldots,x^*_k)$. 

To do so, we first make a change of variables for system~\eqref{subsys}. Recall that $G_r = (V_r,E_r)$ is the subgraph induced by the root set $V_r$. 
We now 
define a set of new variables  $y_i$'s as follows:
\begin{equation*}
y_i :=  \frac{(1-x_i)x_i}{(1-x^*_i)x^*_i}, \hspace{10pt} \forall i\in V_r. 
\end{equation*} 
This is well defined because, from Proposition~\ref{EQUI}, we have $x^*_i\in(0, \alpha]$, and hence the denominator is nonzero. 
We consider below the dynamics of $y_i$, for $i\in V_r$: 
First, we define a set of new coefficients $\widetilde c_{ji}$ as follows:
\begin{equation}\label{widetilde cji}
\widetilde c_{ji} := c_{ji} \frac{(1 - x^*_j)x^*_j}{(1 - x^*_i)x^*_i}, \hspace{10pt} \forall i\to j\in E_r. 
\end{equation} 
For a vertex $i\in V_r$, denote by $V^+_{r,i}$ the set of incoming neighbors of $i$ in $G_r$. It should be clear that $$V^+_{r,i} = V^+_i\cap V_r.$$  
With the notations above, we can express the evolution equations of $y_i$, for $i\in V_r$, as follows:
\begin{equation}\label{ysys1}
\dot{y}_i = (1 - 2x_i) \left (-y_i + \sum_{j\in V^+_{r,i}} \widetilde c_{ji}\, y_j\right ), \hspace{10pt} \forall i\in V_r. 
\end{equation} 

We further describe below some relevant properties of system~\eqref{ysys1}. First, note that $x^*\in Q_r$ is an equilibrium point of system~\eqref{MODEL}, and hence
$$
C^{\top}_{11} (I - X'^*) x'^* = 0,
$$ 
which implies that
\begin{equation}\label{eq:9:4610122015}
(1 - x^*_i)x^*_i = \sum_{j\in V^+_{r,i}} c_{ji} (1 - x^*_j)x^*_j, \hspace{10pt} \forall i\in V_r.  
\end{equation}
Combining~\eqref{widetilde cji} and~\eqref{eq:9:4610122015}, we have that the coefficients $\widetilde c_{ji}$, for $i\to j\in E_r$, satisfy the following condition:
\begin{equation}\label{tildec}
\sum_{j\in V^+_{r,i}} \widetilde c_{ji} = \sum_{j\in V^+_{r,i}} c_{ji} \frac{(1 - x^*_j)x^*_j}{(1 - x^*_i)x^*_i}= 1, \hspace{10pt} \forall i\in V_r. 
\end{equation}
Also, note that $Q_r$ is positive invariant for system~\eqref{MODEL}. Thus, if $x(0)\in Q_r$, then $x(t) \in Q_r$ for all $t \ge 0$; in particular, this implies that
\begin{equation}\label{lowerbd}
1 - 2x_i(t) \ge 1 - 2\alpha >0, \hspace{10pt} \forall \, t \ge 0    \mbox{ and } \forall \, i\in V_r.
\end{equation}
for $\alpha$ defined in~\eqref{alpha}.

Combining~\eqref{ysys1} with~\eqref{tildec} and~\eqref{lowerbd}, we recognize that the dynamics of $y_i$ is a continuous-time consensus process scaled by a time-varying, positive factor $(1 - 2x_i(t))$.   
Now,  define a $k$-by-$k$ infinitesimally stochastic matrix $A(x')= (a_{ij}(x'))$ as follows: for $i\neq j$, let
\begin{equation}\label{eq:A(x)induced}
 a_{ij}(x'): =
\left\{
\begin{array}{ll}
(1 - 2x_i)\widetilde c_{ji}   &  \text{if }j\to i\in E_r\\
0 & \text{otherwise.} 
\end{array}
\right. 
\end{equation} 
The diagonal entries of $A(x')$ are set to be $-(1 - 2x_i)$, for $i\in V_r$. With the matrix $A(x')$, we can rewrite system~\eqref{ysys1} into the following matrix form:
\begin{equation}\label{ysys2}
\dot{y} = A(x') y. 
\end{equation}
With the definitions and notations above, we next prove Proposition~\ref{convQr}:  

\begin{proof}[Proof of Proposition~\ref{convQr}]
We first establish the convergence of system~\eqref{ysys2}, and then show that it implies the convergence of sysem~\eqref{subsys}. 

Fix an initial condition $x(0)\in Q_r$, and let $x(t)$ be the solution of system~\eqref{MODEL}. Then, system~\eqref{ysys2} can be viewed as a time-varying linear consensus process, with $A(x(t))$ (or simply written as $A(t)$) the time-varying infinitesimally stochastic matrix, i.e., 
\begin{equation}\label{ysys3}
\dot y = A(t) y. 
\end{equation}
Moreover, following~\eqref{lowerbd}, we obtain
\begin{equation}\label{lowerbd1}
a_{ij}(t) \ge (1 - 2\alpha)\,\widetilde c_{ji} , \hspace{10pt} \forall t\ge 0 \mbox{ and } \forall j\to i\in E_r. 
\end{equation}
Now, consider the digraph induced by $A(t)$: Let $$\widetilde{G}_r = (V_r,\widetilde E_r)$$ be defined by reversing directions of edges of $G_r$, i.e., $$i\to j\in \widetilde E_r \hspace{5pt} \mbox{ if and only if } \hspace{5pt} j\to i\in E_r.$$ Then, from~\eqref{lowerbd} and~\eqref{eq:A(x)induced}, $\widetilde G_r$ is the digraph induced by $A(t)$. Furthermore, since $G_r$ is strongly connected, so is $\widetilde{G}_r$. In summary, we have that 
\begin{itemize}
\item[1.] The induced digraph of $A(t)$ is $\widetilde G_r$ for all $t\ge 0$, and $\widetilde G_r$ is strongly connected.
\item[2.] Following~\eqref{lowerbd1}, there is a uniform lower bound $\delta > 0$ such that 
$$
a_{ij}(t) \ge \delta, \hspace{10pt} \forall t\ge 0 \mbox{ and } \forall i\to j\in \widetilde E_r.
$$
\end{itemize}
Then, it is known from \cite{luc04} that the transition matrix $\Phi(t)$ of the linear system~\eqref{ysys3} converges to a rank one matrix:   
\begin{equation*}
\lim_{t\to\infty}\Phi(t) = \mathbf{1} \cdot u^\top 
\end{equation*}
where $\mathbf{1}\in \bb{R}^k$ is a vector of all ones, and $u$ is a vector in the $(k-1)$-simplex. Hence, the solution $y(t)$ also converges: 
\begin{equation*}
\lim_{t\to \infty}y(t) = \langle u, y(0)\rangle \cdot \mathbf{1}. 
\end{equation*}
It suffices to show that the convergence of $y(t)$ implies the convergence of $x'(t)$ of system~\eqref{subsys}. To see this, first note that $Q_r$ is positive invariant, and hence, 
\begin{equation*}
x_i(t) \in [0,\alpha], \hspace{10pt} \forall t \ge 0 \mbox{ and } \forall i\in V_r. 
\end{equation*} 
Since $\alpha< 1/2$,  the function $(1 - x)x$ is invertible when restricted to the closed interval $[0,\alpha]$; indeed, $(1 - x)x$ is strictly monotonically  increasing when restricted to $[0,\alpha]$.  
Thus, if $y_i(t) = (1 - x_i(t))x_i(t)$  converges, then so does $x_i(t)$, for all $i\in V_r$.  On the other hand, from Proposition~\ref{EQUI}, there is only one equilibrium point $x^*$ in $Q_r$; we thus conclude that $x'(t)$ converges to $x'^*$, and hence $x(t)$ converges to $x^*$. This completes the proof. 
\end{proof}

\subsection{Proof of Theorem~\ref{MAIN}}\label{ssec:VD}

We provide in this subsection a complete proof of Theorem~\ref{MAIN}. We first recall that $\alpha$ is a scalar (defined in~\eqref{alpha}) given by
$$
\alpha = \frac{1}{2} - \frac{1}{4(2n - 3)},
$$
We also recall that $Q$, $Q_r$ are subsets of $\Sp[V]$, given by
$$
Q = \{x \in \Sp[V] \mid x_i \le \alpha, \, \forall\, i\in V\}
$$
and 
$$
Q_r = Q\cap \Sp[V_r].
$$
To prove Theorem~\ref{MAIN}, we further need the following: for a vector $v = (v_1,\ldots, v_n)\in \R^n$, let the {\it one-norm} of $v$ be defined as follows: $$\|v\|_1 := \sum^n_{i = 1} |v_i|;$$
then, the following holds:

{ \color{black}
\begin{lem}\label{lem:11:28am}
Let $x$ be in $Q$ with $\sum_{i\notin V_r} x_i \le \epsilon$, for $\epsilon \le 1/4$; then, there exists a point $x'$ in $Q_r$ such that
$$
\|x - x'\|_1 \le 2\epsilon. 
$$ \, 
\end{lem}

\begin{proof}
We first show that 
\begin{equation}\label{eq:11:49am10122015}  
\sum_{i\in V_r} (\alpha - x_i) \ge 1/4 
\end{equation}
This holds because $\sum_{i\in V_r} x_i \le 1$, and 
$$
\sum_{i\in V_r} \alpha \ge 3\alpha = \frac{3}{2} - \frac{3}{4(2n - 3)}\ge \frac{5}{4}. 
$$
The first (resp. last) inequality holds because $|V_r| \ge 3$ (resp. $n \ge 3$). Since $\epsilon \le 1/4$, from~\eqref{eq:11:49am10122015}, there exists a vector $x' \in Q_r$ such that $x'_i \ge x_i$ for all $i \in V_r$, and moreover,
$$
\sum_{i\in V_r} (x'_i - x_i) = \sum_{i\notin V_r} x_i. 
$$
It then follows that 
$$
\|x' - x \|_1 = 2\sum_{i\notin V_r} x_i \le 2\epsilon,
$$
which completes the proof.
\end{proof}
}
We are now in a position to prove Theorem~\ref{MAIN}:  

\begin{proof}[Proof of Theorem~\ref{MAIN}]
Parts~1 and~2 of Theorem~\ref{MAIN} are established by Corollary~\ref{spvvv} and Proposition~\ref{EQUI}, respectively. We prove here Part~3.

First, we show that there is an open set $U$ in $\Sp[V]$ containing $Q_r$ such that if $x(0) \in U$, then the solution $x(t)$ of system~\eqref{MODEL} converges to $x^*$. 
By Proposition~\ref{EQUI}, the non-vertex equilibrium point $x^*$ is exponentially stable. Thus, there exists an open neighborhood $U_{x^*}$ of $x^*$ in $\Sp[V]$ such that if $x(0)\in U_{x^*}$, then 
$x(t)$ converges exponentially fast to $x^*$. 
Now choose a point $x\in Q_r$. From Proposition~\ref{convQr}, we know that there exists an instant $T\ge 0$ such that if $x(0) = x$, then $x(t)\in U_{x^*}$ for all $t\ge T$. This, in particular, implies that there exists an open neighborhood $U_{x}$ of $x$ in $\Sp[V]$ such that if $x(0) \in U_{x}$, then $x(T) \in U_{x^*}$.    
Now, define
\begin{equation*}
U := \cup_{x\in Q_r} U_{x}.
\end{equation*}
Then, $U$ is the desired open set of $Q_r$ in $\Sp[V]$. 

{ \color{black}
Next, we show that there exists an $\epsilon > 0$ such that if $x$ lies in $Q$ and satisfies 
$
\sum_{i\notin V_r} x_i \le \epsilon 
$, then $x \in U$.  
Since $Q_r$ is a compact subset and $U$ is an open set containing $Q_r$, there exists an $\epsilon' > 0$ such that the {\it $\epsilon'$-neighborhood} of $Q_r$ in $\Sp[V]$ is contained in $U$. Specifically, let 
$$
U' := \{ x \in \Sp[V]  \mid \|x - x'\|_1 < \epsilon' \mbox{ for some } x'\in Q_r\};  
$$ 
then, $U' \subset U$. Now, choose $\epsilon$ positive but sufficiently small such that
$$
\epsilon \le \min\{\epsilon'/2, 1/4\}.
$$ 
Then, from Lemma~\ref{lem:11:28am}, we have that for any $x\in Q$  with $\sum_{i\notin V_r} x_i \le \epsilon$, there exists a $x'\in Q_r$ such that
$$
\|x - x'\| \le 2\epsilon \le \epsilon',
$$
which implies that $x\in U'\subset U$. 

We now show that any trajectory $x(t)$ of system~\eqref{MODEL} intersects the open set $U$. Note that if this is the case, then by using the arguments as before, we know that $x(t)$ converges to $x^*$. 
To establish the statement, we first note that from Corollary~\ref{DefQ}, there is an instant $T_1$ such that $x(t)\in Q$ for all $t\ge T_1$. We also note that from Proposition~\ref{pro31}, there exists an instant $T_2$ such that 
$$
\sum_{i\notin V_r} x_i(t) \le \epsilon, \hspace{10pt} \forall t\ge T_2.  
$$
Thus, if we let $T:= \max\{T_1,T_2\}$, then $x(t) \in Q$ and $\sum_{i\notin V_r} x_i(t) \le \epsilon$ for all $t \ge T$, which implies that 
$$
x(t) \in U, \hspace{10pt} \forall t\ge T. 
$$
This completes the proof. 
}
\end{proof}

\section{Conclusions}
In this paper, we have introduced a continuous-time model whereby a number of agents in a social network are able to evaluate their self-appraisals over time in a distributed way. We have shown that under a particular assumption (Assumption~\ref{asmp:keyasmp}),  the solution of the system converges to the unique non-vertex equilibrium point $x^*$ as long as the initial condition is not a vertex, and moreover, $x^*$ is exponentially stable. This stable equilibrium point can be interpreted as the steady state of the self-appraisals of the agents, and we have related the value of each $x^*_i$ to the  values of $c_{ji}$,  as  described by Proposition~\ref{VR} and Corollary~\ref{R1}. 

Future work may focus on the case where each $c_{ij}$ is time-variant. For example, we can assume that each $c_{ij}$ also depends on the self-appraisal of $x_j$, e.g., how much opinion agent~$i$ accepts from agent~$j$ now depends on how influential agent~$j$ is in the social network. Also we note that this model can be developed into many other interesting problems. For example, a question related to sparse systems is that given the underlying graph $G$, what is the collection of achievable steady states $x^*$ by the choice of $c_{ij}$'s? A similar question has been asked and answered in the context of the consensus process~\cite{chen2015consensus}. 
Another problem related to optimal control is to assume that there is an agent~$i$ who is able to manipulate her own weights $c_{ij}$, and we ask whether there is a choice of these weights so that self-appraisal of agent~$i$ is maximized? If further, we assume that there are two such players each of whom is trying to maximize her own self-appraisal, then what would be the strategy for each of the players to choose the $c_{ij}$'s? This list of examples indicates that the self-appraisal model has a rich structure which can be investigated under various assumptions and from different perspectives.

\bibliographystyle{unsrt}
\bibliography{ji,CDC2015}

\section*{Appendix}

We prove here Proposition~\ref{pro2}. We first have some definitions. 
A polynomial $m(x_1,\ldots, x_n)$ is said to be a {\bf monomial} if it can be expressed as follows:
 \begin{equation*}
m(x_1,\ldots, x_n)=\Pi^n_{i=1}x^{k_i}_i.
 \end{equation*} 
 with $k_i\ge 0$ for all $i=1,\ldots,n$. The {\bf degree} of the monomial $m(x_1,\ldots, x_n)$ is defined to be $\sum^n_{i=1} k_i$.  A polynomial $p(x_1,\ldots, x_n)$ can be uniquely expressed as a linear combination of monomials. We say a monomial is {\it contained in $p$} if the corresponding coefficient is nonzero. Furthermore, the degree of $p$ is defined to be the highest degree of a monomial contained in $p$. 
With these definitions at hand, we have the following result:

\begin{lem}\label{prolem24}
Let $x(0)$ be any initial condition of system~\eqref{MODEL} in $\Sp[V]$. The time derivative $x_i^{(k)}(0)$ is a degree-$(k+1)$  polynomial in $x_j(0)$, for $j\in S_i(k)$,  and it does not contain a constant term. Moreover, if $D_i(k)$ is nonempty for $k\ge 1$, then the following properties hold:  
\begin{enumerate}
\item 
The only monomials in variables $x_j(0)$, for $j\in D_i(k)$, contained in $x^{(k)}_i(t)$ are 
$$x_j(0) \hspace{5pt} \mbox{ and } \hspace{5pt}  x^2_j(0), \hspace{10pt} \forall\,  j\in D_i(k) .$$
\item The coefficients  of $x_j(0)$ and $x^2_j(0)$, for $j\in D_i(k)$, in $x^{(k)}_i(0)$ are $\alpha_{ji} $ and $-\alpha_{ji} $,   respectively. 
\end{enumerate}\,
\end{lem}

\begin{proof}
In the proof, we will simply write $x^{(k)}_i$ by omitting the argument. 
The proof will be carried out by induction on $k$.

{\it Base case}. We assume that $k = 1$;  in this case, we have  
$$
x^{(1)}_i = -(1 - x_i)x_i + \sum_{j\in V^+_i} c_{ji}(1 - x_j) x_j.
$$
We recall the fact that 
$$ 
S_i(1) = V^+_i\cup\{i\} \hspace{10pt} \mbox{ and } \hspace{10pt} D_i(1) = V^+_i,$$ 
and the fact that  $$\alpha_{ji} = c_{ji}, \hspace{10pt} \forall \,j\in V^+_i.$$ 
It then follows that Lemma~\ref{prolem24} holds for $k = 1$.

{\it Inductive step}. We assume that Lemma~\ref{prolem24} holds for $(k-1)$, and prove for $ k $. We first show that $x^{(k)}_i$ is a degree-$(k+1)$ polynomial in variables $x_{j}$'s, for $j\in S_i(k)$, and it does not contain a constant term. 

Label the elements of $S_i(k-1)$ as $i_1,\ldots, i_s$. 
By the induction hypothesis, we can write $x^{(k-1)}_i$ as follows:
\begin{equation*}
x^{(k-1)}_i = \sum^t_{q=1} \sigma_q m_q. 
\end{equation*}
Each $m_q$ is a monomial in variables $x_{i_1},\ldots, x_{i_s}$, and $\sigma_q$ is the associated coefficient. The degree of $m_q$ is at most $k$. 
By chain rule, we have
\begin{equation}\label{chainrule}
 x^{(k)}_i = \sum^t_{q=1} \sum^s_{p=1} \sigma_q\frac{\partial m_{q}}{\partial x_{i_p}} f_{i_p}(x).
\end{equation}
Note that each $f_{i_p}(x)$ (defined in~\eqref{eq:fi(x)}) is a degree-$2$ polynomial in variables $x_j$, for  $j\in S_{i_p}(1)$, and it does not contain a constant term. Also, note that since $i_p\in S_i(k-1)$,  each $ S_{i_p}(1) $ is a subset of $S_i(k)$. Thus, by following~\eqref{chainrule}, we know that (i) $x^{(k)}_i$ is a polynomial in variables $x_j$, for $j\in S_i(k)$; (ii) the degree of $x^{(k)}_i$ is at most $(k+1)$; and (iii) $x^{(k)}_i$  does not contain a constant term. 
Further, a direct computation yields that the monomial $x^{k+1}_i$ is contained in $x^{(k)}_i$ with coefficient $k!$.  Thus, the degree of $x^{(k)}_i$ is $(k+1)$.

We now assume that $D_i(k)$ is nonempty, and establish items~1 and~2 of the lemma.
First, note that each ${\partial m_{q}}/{\partial x_{i_p}}$ is a monomial in $x_{j}$, for $j\in S_i(k-1)$. Hence, if $({\partial m_{q}}/{\partial x_{i_p}}) f_{i_p}(x)$ contains a monomial in  $x_{j}$, for $j\in D_i(k)$;  then, we {\it must} have 
\begin{equation}\label{eq:10:03pm10102015}
m_q = x_{i_p} 
\end{equation}
so that ${\partial m_{q}}/{\partial x_{i_p}} = 1$,  
and moreover, $f_{i_p}(x)$ {\it has to} contain a monomial in  $x_{j}$, for $j\in D_i(k)$. 
On the other hand, $f_{i_p}(x)$ is a polynomial in variables $x_j$, for $j\in S_{i_p}(1)$. From Lemma~\ref{prolem25}, if $S_{i_p}(1)$ intersects $D_i(k)$, then
\begin{equation}\label{eq:10:04pm10102015}
i_p \in D_{i}(k-1).
\end{equation}
Combining~\eqref{eq:10:03pm10102015} and~\eqref{eq:10:04pm10102015}, we know that $x^{(k)}_i$ contains at most $x_{j}$ and $x^2_j$, for $j\in D_i(k)$,   as the monomials in $x_j$, for $j\in D_i(k)$.

Now, fix a vertex $j\in D_i(k)$, and let $\sigma'_{j}$ and $\sigma''_j$ be the coefficients of $x_j$ and $x^2_j$ in $x^{(k)}_i$. It suffices to show that $$\sigma'_j = \alpha_{ji} \hspace{10pt} \mbox{ and }  \hspace{10pt} \sigma''_j = -\alpha_{ji}.$$ 
First, from Lemma~\ref{prolem25},  $V^-_j$ intersects $S_i(k-1)$, and their intersection lies in $D_i(k-1)$. We thus define
$$
W: = D_i(k-1) \cap V^-_j.
$$
Then, using the fact that $({\partial m_{q}}/{\partial x_{i_p}}) f_{i_p}(x)$ contains the monomials $x_j$ and $x^2_j$ if and only if $m_q = x_{i_p}$ and $i_p \in W$, we obtain that  
$$
\sigma'_j = -\sigma''_j = \sum_{i_p\in W} c_{ji_p} \sigma_{i_p},  
$$
where $\sigma_{i_p}$ is the coefficient of $x_{i_p}$ in $x^{(k-1)}_i$. 
From the induction hypothesis, we have $\sigma_{i_p} = \alpha_{i_pi}$, and hence 
$$
\sigma'_j = -\sigma''_j = \sum_{i_p\in W} c_{ji_p} \alpha_{i_pi}, 
$$
which is $\alpha_{ji}$ by~\eqref{eq:recursivedefinitionofalphaji}. 
This completes the proof. 
\end{proof}

With Lemma~\ref{prolem24} at hand, we now prove Proposition~\ref{pro2}:

\begin{proof}[Proof of Proposition~\ref{pro2}]
We first prove for the case where there does not exist a supporting path for $i$ at $x(0)$.  
The fact that $i\notin V_r$ follows from Lemma~\ref{lempro23}. To see this, we assume that $i\in V_r$. But then, $S_i = V$, and hence, $x_j(0) = 0$ for all $j\in V$, which contradicts the fact that $x(0)\in \Sp[V]$. 
We now show that 
\begin{equation}\label{solsbsys}
x_j(t) = 0, \hspace{10pt} \forall t\ge 0 \mbox{ and } \forall j\in S_i.
\end{equation}
Note that the dynamics of $x_j$ depends only on $x_k$'s for $k\in S_j(1)$. Since $j\in S_i$, we have  $S_j(1)\subseteq S_i$. This, in particular, implies that the following system: 
\begin{equation}\label{eq:subsyssj}
\dot{x}_j = -(1 - x_j)x_j + \sum_{k\in V^+_j} c_{kj}(1 - x_k)x_k, \hspace{10pt} \forall j\in S_i
\end{equation}
formed by $x_j$'s, for $j\in S_i$, is an isolated system, independent of $x_{j'}$, for $j'\notin S_i$.  Since there is no supporting path for $i$ at $x(0)$,  by Lemma~\ref{lempro23}, we have 
$$x_j(0) = 0, \hspace{10pt} \forall j\in S_i.$$ 
But then, following this initial condition,~\eqref{solsbsys} is the unique solution of system~\eqref{eq:subsyssj}.

We now prove for the case when there exists a critical supporting path for $i$ at $x(0)$ of length $k$. We first show that $x^{(l)}_i(0) = 0$ for all $l <k$, and then show that $x^{(k)}_i(0)>0$. By Lemma~\ref{prolem24}, $x^{(l)}_i(0)$ is a polynomial in $x_j(0)$, for $j\in S_i(l)$, and it does not contain a constant term. Since the length of a critical supporting path for $i$ at $x(0)$ is $k$, we have
$ x_j(0) = 0$ for all $j\in S_i(l)$. 
Thus,  $x^{(l)}_i(0) = 0$ for all $l<k$.

We now show that $x^{(k)}_i(0)>0$. Since there is a critical supporting path for $i$ at $x(0)$, $D_i(k)$ is nonempty. Then, by Lemma~\ref{prolem24}, we have 
\begin{equation*}
x^{(k)}_i(0) = \sum_{j\in D_i(k)} \alpha_{ji} (x_j(0) - x^2_j(0)).
\end{equation*} 
Since $x(0)$ is not a vertex of $\Sp[V]$, we have $x_j(0)<1$ for all $j\in V$. On the other hand, there exists at least a vertex $j\in D_i(k)$ such that $x_{j}(0)>0$. This holds because there exists a critical supporting path for $i$ at $x(0)$ of length $k$.  It then follows that  
$
x^{(k)}_i(0) >0
$, which completes the proof.
\end{proof}

\end{document}